\def\comments{1}

\documentclass[letterpaper,11pt]{article}
\usepackage[utf8]{inputenc}
\usepackage{comment}
\usepackage{preamble}
\allowdisplaybreaks

\newcommand{\nope}[1]{}
\usepackage{titling}

\newcommand{\abs}[1]{\left| #1 \right|}

\renewcommand{\epsilon}{\varepsilon}

\newcommand{\id}{\mathbb{I}}
\newcommand{\Lap}{\mathsf{Lap}}

\newcommand{\SD}{\ensuremath{\mathrm{d_{TV}}}}

\newcommand{\PPDE}{\ensuremath{\textsc{PureDPPDE}}}
\newcommand{\DPSGLearner}{\ensuremath{\textsc{DPSGLearner}}}

\newcommand{\trunc}{\mathsf{trunc}}
\newcommand{\tmean}{\mathsf{tmean}}
\newcommand{\Ber}{\ensuremath{\mathrm{Ber}}}
\newcommand{\getsr}{\gets_{\mbox{\tiny R}}}

\newcommand{\sd}[2]{\SD\left( #1 , #2 \right)}
\newcommand{\KL}{\ensuremath{\mathrm{d_{KL}}}}

\newcommand{\CS}{\ensuremath{\mathrm{d}_{\chi^2}}}
\newcommand{\CSD}{\ensuremath{\chi^2}}

\newcommand{\MI}{\mathit{I}}

\usepackage[style=alphabetic,backend=bibtex,maxalphanames=10,maxbibnames=20,maxcitenames=10,giveninits=false,doi=false,url=true,backref=true]{biblatex}
\newcommand*{\citet}[1]{\AtNextCite{\AtEachCitekey{\defcounter{maxnames}{2}}}\textcite{#1}}

\newcommand*{\citep}[1]{\citep{#1}}

\usepackage{hyperref}

\title{A Polynomial Time, Pure Differentially Private Estimator\\for Binary Product Distributions
}
\iftrue
\author{
    Vikrant Singhal
    \thanks{\texttt{vikrant.singhal@uwaterloo.ca}. Cheriton School of Computer Science, University of Waterloo. Supported by an NSERC Discovery Grant.}}
\date{}
\fi

\begin{document}
\maketitle
\thispagestyle{empty}

\begin{abstract}
    We present the first $\eps$-differentially private, computationally efficient algorithm that estimates the means of product distributions over $\zo^d$ accurately in total-variation distance, whilst attaining the optimal sample complexity to within polylogarithmic factors. The prior work had either solved this problem efficiently and optimally under weaker notions of privacy, or had solved it optimally while having exponential running times.
\end{abstract}




\section{Introduction}

Machine learning and statistics aim to learn information about the population. The pertinent algorithms always involve using random samples from the relevant population to learn and release that information, but at the same time, often end up revealing sensitive information about the individuals in the datasets. \emph{Differential privacy (DP)} \cite{DworkMNS06} is now a \emph{de facto} standard for preserving privacy in learning and testing algorithms. It informally guarantees that no adversary can infer anything more about an individual from the output of a differentially private algorithm, than they could have from its output if the individual were not present in the dataset.

In the last few years, a large body of work on differentially private statistics has emerged, which has shown that the privacy constraint almost always imposes an additional cost in the sample complexity for those tasks (see Section~\ref{sec:intro-related}). Depending on the strength of the privacy guarantee (e.g., pure, concentrated \cite{BunS16, DworkR16}, or approximate differential privacy), the running times of the private algorithms also tend to get affected greatly. It is quite often (but \emph{not} always) the case that developing pure DP algorithms for a statistical task, which has the optimal sample complexity and a polyomial running time, is much more challenging than creating efficient and (sample) optimal, concentrated or approximate DP algorithms for the same task. For instance, mean estimation of heavy-tailed distributions in $\ell_2$ distance under pure DP was solved optimally, but without computational efficiency, in \cite{KamathSU20}, but obtaining a computationally efficient algorithm for the same remained an open problem until it was solved recently in \cite{HopkinsKM22}. On the other hand, \cite{KamathSU20} did provide computationally efficient and optimal, concentrated and approximate DP algorithms for mean estimation of heavy-tailed distributions. Similarly, \cite{KamathLSU19} presented efficient and optimal, concentrated and approximate DP algorithms for mean estimation of multivariate Gaussians, but solving this problem optimally and with computational efficiency under pure DP remained open until \cite{HopkinsKMN23} solved it recently.

Estimating binary product distributions (product distributions over $\zo^d$) in total-variation distance is another such example. In this case, too, optimal and computationally efficient, concentrated and approximate DP algorithms had been presented already in \cite{KamathLSU19}, but no computationally efficient and optimal, pure DP algorithm was known after that, although multiple optimal algorithms for this problem under pure DP, but lacking the computational efficiency, have come up in the recent past (e.g., \cite{BunKSW19}). In this work, we provide an optimal and a computationally efficient algorithm that satisfies pure DP, and close that gap for this problem.

\subsection{Estimation in Total-Variation Distance}
    \label{sec:intro-challenges}

We would first like to remind the reader as to why estimating binary product distributions in total-variation distance is a much more difficult task than just simply estimating them in $\ell_2$ distance. A small error in total-variation distance requires all the marginals to be estimated accurately with respect to the magnitudes of their respective means, that is, the error for each marginal needs to be scaled according to the magnitude of its mean. In certain situations, it could imply that each marginal needs to be estimated to within a small multiplicative error. Instead, if in the estimate, all the coordinates have similar additive errors, then unless that error is \emph{very} small, the marginals with very small means would have significantly lower accuracy than the ones with much larger means. Therefore, the estimate would not be accurate in total-variation distance. On the other hand, if we have the additive errors to be very small for all the coordinates (say, if we are estimating the distribution to within $\ell_1$ distance $\alpha$ or to within $\ell_{\infty}$ distance $\tfrac{\alpha}{d}$), then the cost of estimation would become very large in terms of the sample complexity. Therefore, we need to find a way to estimate product distributions accurately \emph{direction-wise}.

Sans privacy, this is an easy task -- we can just output the empirical mean of the samples. Under privacy constraints, however, this is hard to do without knowing even an approximate scale of the noise to add in each direction. The na\"{i}ve way to estimate privately would just involve either adding a very small noise to each coordinate of the empirical mean, but that would increase the cost in the sample complexity dramatically (by $\poly(d)$). On the other hand, even if we use a sophisticated pure DP estimator, which simply estimates accurately in $\ell_2$ distance $\alpha$, this would not be accurate enough either. Thus, we need more non-trivial ways to privately estimate each coordinate to within an error that is scaled appropriately for that marginal.

\subsection{Result}

We informally state the main result of this work here. It essentially says that our pure DP algorithm is computationally efficient and estimates the mean of any product distribution over $\zo^d$ in total-variation distance using just $\wt{O}(d)$ samples.

\begin{thm}[Informal]\label{thm:ppde-informal}
    For every $\eps,\alpha,\beta > 0$, there exists a
    polynomial-time, $\eps$-DP algorithm that takes $n$ i.i.d.\ samples from a product distribution $P$ over $\zo^d$, and returns a product distribution $Q$, such that if
    $$n \geq \wt{O}\left(
        \frac{d}{\alpha^2} + \frac{d}{\eps\alpha}\right),$$
    where $\wt{O}(\cdot)$ hides all polylogarithmic factors,
    then with probability at least $1-\beta$, the total-variation distance between $P$ and $Q$ is at most $\alpha$.
\end{thm}

Note that this is the optimal sample complexity for this problem, and the upper bounds under approximate DP and the matching lower bounds (which also hold under pure DP) were proved in \cite{KamathLSU19}. The first term in the sample complexity is the necessary term for attaining the desired accuracy without any privacy constraints, and the second term is the additive cost due to privacy. For $\eps \geq \Omega(\alpha)$, privacy comes for ``free'', that is, there is only a small multiplicative cost in the non-private sample complexity. Additionally, as we will see in the formal version of the above theorem (see Theorem~\ref{thm:ppde}), there is also a multiplicative $\polylog(d,\tfrac{1}{\beta},\tfrac{1}{\eps})$ improvement in the sample complexity over that of \cite{KamathLSU19} due to our tighter analysis.

One more comparison that we would like to draw is with the work of \cite{BunKSW19}. It is true that the polylog factors in their results are better than ours, but their algorithm only provides guarantees for estimation within total-variation distance, whereas our work does that for parameter estimation within $\chi^2$-distance (hence, KL-divergence, as well) for a very wide range of parameters. In that sense, our algorithm provides stronger results because the total-variation distance guarantee is implied from the above.

\subsection{Overview of Techniques}\label{sec:intro-techniques}

Our techniques are similar to those in \cite{KamathLSU19} for estimating binary product distributions -- private partitioning, followed by privately estimating.

The private partitioning is performed iteratively. In each round, we assume an upper bound on the marginals, and use the Laplace mechanism (scaled according to that upper bound and the number of marginals that remain to be partitioned) to get a rough private estimate of each marginal, and pick the ones that lie above a certain threshold. As we prove later, those coordinates are bound to have higher means than the ones with their noisy estimates below that threshold. We rescale those chosen heavier coordinates as per the assumed upper bound for that iteration, and mark them to be estimated later after the partitioning is complete. In the next round, we assume both a reduced upper bound on the marginals and a reduced threshold to filter out the next set of marginals, and repeat the process.

Once we have filtered and rescaled those heavier marginals, we estimate them in $\poly(n,d)$ time to within $\alpha$ in $\ell_2$ distance under pure DP using the sub-Gaussian learner from \cite{HopkinsKMN23}. On inverse-rescaling the estimated marginals according to their respective original rescaling parameters, we get an accurate estimate for those heavier coordinates to within $O(\alpha)$ in total-variation distance. Note that we cannot simply invoke the estimator from \cite{HopkinsKMN23} on the filtered out coordinates before rescaling them, since that would just give an $\ell_2$ estimate of the original marginals, which would not be accurate direction-wise, hence, would not be accurate in total-variation distance, especially when there are marginals with very different magnitudes in that filtered set (see Section~\ref{sec:intro-challenges}). Therefore, this combination of partitioning and rescaling those heavier marginals seems necessary for this kind of an approach.

\begin{rem}
    We would like to remark that in \cite{HopkinsKMN23}, the assumption is that the covariance of the distribution in question is $\Sigma = \id$, while what we require is that $\Sigma \preceq \id$. That said, their algorithm can still work with the latter assumption because their proof for mean estimation mainly relies on their Corollary~5.4 and Lemma~B.1, which still hold under this relaxed assumption.
\end{rem}

There are two important aspects of our algorithm that make it different from the work by \cite{KamathLSU19}.
\begin{itemize}
    \item In the partitioning procedure, we filter out the ``heavier'' marginals iteratively, and group the ones with similar weights together, but unlike the algorithm in \cite{KamathLSU19}, ours does not estimate them simultaneously while the partitioning is being performed. This is because we wanted to avoid the additional $\poly(d)$ cost in the sample complexity due to (basic) composition under pure DP. However, \cite{KamathLSU19} were able to both partition and estimate those heavy marginals at the same time because they were working under concentrated or approximate DP, and they could use the more sophisticated, advanced composition of privacy \cite{DworkRV10}, which only provides concentrated or approximate DP guarantees.
    \item In order to estimate those heavier marginals after filtering and grouping them, we rescale them using their respective rough private estimates that we obtained while partitioning, and apply the pure DP sub-Gaussian mean estimator from \cite{HopkinsKMN23}, which is computationally efficient, as well. This gives us an estimate of the scaled marginals that is accurate to within $\alpha$ in $\ell_2$ distance.
\end{itemize}

After the partitioning rounds, we have the final round, where we are just left with the lighter coordinates to estimate. For that, we simply use the Laplace mechanism again, but with a much lower sensitivity this time, and get an accurate estimate for those marginals directly in one shot.

With accurate estimates for both the heavy and the light marginals in hand, we finally combine the two via a simple concatenation, and this gives us an accurate estimate for the whole distribution, as we had originally desired.

\begin{rem}
    We also want to remark on the techniques of \cite{BunKSW19} for comparison. The algorithm of \cite{BunKSW19} involves an application of the \emph{exponential mechanism} \cite{McSherryT07}, but it is more than just that. It is also a very general-purpose algorithm. While simpler applications of the exponential mechanism could be made computationally efficient, doing that for the algorithm in \cite{BunKSW19} for our case does not seem straightforward. Intuitively, in their work, the score function of a point in the output space is also based on its tournaments with all the other points in the output space (as opposed to just with respect to the dataset), which means that sampling a point from that space via the exponential mechanism becomes inefficient as there is no way to efficiently determine the score of that point itself.
\end{rem}

\subsection{Motivation and Broader Impact}

We would like to remark that in this recent line of work on computationally efficient, pure DP statistical estimation, our work deals with a very fundamental distribution under a tricky, direction-wise error metric (total-variation distance). Therefore, while it closes a long-standing open problem, it also contributes to the diverse body of algorithmic tools and ideas available in DP literature for statistical estimation tasks. Could the techniques in our work be directly applied to other distributions? The answer is not obvious because different families of distributions have different properties and the way distance metrics are characterised for them could be very different. However, the high-level idea of performing private preconditioning can be and has been certainly useful in DP statistical estimation tasks. Also, we believe that it might be possible to estimate certain families of distributions (say, a subset of those with finite domains) in metrics that are similar to $\CSD$-divergence using ideas from our algorithm, just like we do in this work.

Additionally, our work also shows that even though heavy and general-purpose machinery, such as \cite{HopkinsKMN23}, might be available at our disposal, many tasks (like the problem we address) may still not have simple solutions. As we point out in Section~\ref{sec:intro-techniques}, a lot of non-trivial steps often need to be taken in order to effectively use these tools.

Computationally efficient statistical estimation under pure DP has been a topic of recent interest in the DP community, mostly because many pure DP algorithms, even though they may have optimal sample complexity, are not very practical as they tend to have exponential running times. Pure DP gives us much stronger privacy guarantees, and if we have as practical and sample-efficient algorithms as the ones under the weaker privacy notions (such as approximate DP) to solve problems under this regime, then we have the best of both worlds. Hence, our goal was also to fill in another gap in this literature and provide new algorithms in this broader line of work for the community interested in DP statistical estimation.

Estimating binary product distributions is a fundamental statistical problem, so we answered it in the contexts (1) of computationally efficient pure DP estimation and (2) of addressing a folklore statistical problem under DP constraints. Therefore, we believe that we addressed this question from both theoretical and practical perspectives.

\subsection{Related Work}\label{sec:intro-related}

Besides a large selection of folklore work in non-private estimation of distributions, there has been a lot of work in recent years on differentially private statistical estimation.
Mean estimation is possibly the most fundamental question in this space, enjoying significant attention (e.g.,~\cite{BarberD14, DuchiJW13, KarwaV18, BunS19, KamathLSU19, KamathSU20, WangXDX20, DuFMBG20, BiswasDKU20, CaiWZ21, BrownGSUZ21, HuangLY21, LiuKKO21, LiuKO22, KamathLZ22, HopkinsKM22, KothariMV22, TsfadiaCKMS22, DuchiHK23, CovingtonHHK21, NikolovT23, KamathMRSSU23}).
Other related problems include private covariance or density estimation~\cite{BunNSV15, BunKSW19,AdenAliAK21,KamathMSSU22, AshtianiL22, AlabiKTVZ22,HopkinsKMN23, KothariMV22, TsfadiaCKMS22, LiuKO22, BiswasDKU20}.
Beyond these settings, other works have examined statistical estimation under differential privacy constraints for mixtures of Gaussians~\cite{KamathSSU19, AdenAliAL21,ChenCDEIST23}, graphical models~\cite{ZhangKKW20}, discrete distributions~\cite{DiakonikolasHS15}, median estimation~\cite{AvellaMedinaB19, TzamosVZ20,RamsayC21,RamsayJC22,BenEliezerMZ22, CummingsD20}, and more.
Several recent works have explored the connections between privacy and robustness \cite{LiuKKO21, HopkinsKM22, GeorgievH22, LiuKO22, KothariMV22, AlabiKTVZ22, HopkinsKMN23,ChenCDEIST23}, and between privacy and generalization \cite{HardtU14,DworkFHPRR15,SteinkeU15,BassilyNSSSU16,RogersRST16,FeldmanS17}.
Upcoming directions of interest include ensuring privacy when one individual may contribute multiple data points \cite{LiuSYKR20, LevySAKKMS21, GeorgeRST22} (or what is known as, \emph{user-level differential privacy}), a combination of local and central DP for different users~\cite{AventDK20}, and estimation with access to trace amounts of public data~\cite{BieKS22}.
We refer the reader to \cite{KamathU20} for more coverage of the recent work on differentially private statistical estimation.
Differentially private statistical inference also has been an active area of research for over a decade (e.g., \cite{DworkL09, VuS09, WassermanZ10, Smith11}), but the literature is too broad to fully summarize here.

Another broader line of work includes those on the minimax sample complexities for various differentially private statistical estimation tasks.
The first minimax sample complexity bounds to show an asymptotic separation between private and non-private estimation for private mean estimation were proved in \cite{BunUV14}, and subsequently sharpened and generalized in several ways \cite{DworkSSUV15,BunSU17,SteinkeU17a,SteinkeU17b,KamathLSU19}. Recently, \cite{CaiWZ21} extended these bounds to sparse estimation and regression problems. \cite{AcharyaSZ21} provides an alternative, user-friendly approach to proving sample complexity bounds, which is directly analogous to the classical approaches in statistics and learning theory for proving minimax lower bounds.
\cite{KamathMS22} also provides a generalized version of the well-known ``fingerprinting technique'' that is used to prove hardness results for these kinds of problems under approximate DP.

The two prior works most relevant to ours among the above are those by \cite{KamathLSU19, HopkinsKMN23} on learning binary product distributions under concentrated and approximate DP, and estimating means of sub-Gaussian distributions under pure DP, respectively. The question of learning binary product distributions optimally in polynomial time under pure DP has stayed open for a while now. In our work, we modify and use a combination of their techniques to solve this problem.

\section{Preliminaries}

For the utility analysis, we mostly rely on the concentration properties of Bernoulli distributions and the Laplace distribution, and on the relationships among different distance metrics for probability distributions. The privacy analysis is much simpler, and we use the privacy guarantees of the existing differentially private mechanisms. We describe all these notions and results in this section.

\subsection{Statistics Preliminaries}

Here, we state the essential definitions and results from
statistics that would be used throughout the draft.
They include descriptions of various metrics for
probability distributions, and a few useful concentration
inequalities.

\paragraph{Notations.} Let $P_1,\dots,P_k$ be distributions over domains $\cX_1,\dots,\cX_k$, respectively. Then we say that $P=P_1\otimes\dots\otimes P_k$ is a product distribution over $\cX_1\otimes\dots\otimes\cX_k$. Next, for a distribution $P$ over domain $\cX$, we use $P^{\otimes k} = P\otimes\dots\otimes P$ ($k$ times) to denote the product distribution over $\cX\otimes\dots\otimes\cX$ ($k$ times), where each marginal is $P$.
Also, for $0\leq p \leq1$, we use $\Ber(p)$ to denote a Bernoulli random variable with mean $p$.
Finally, for any $v=(v_1,\dots,v_d) \in \R^d$, we use $\abs{v}$ to denote its $\ell_1$ norm, i.e., $\abs{v} = \sum\limits_{i=1}^{d}{\abs{v_i}}$.

\subsubsection{Distances Between Distributions}

We use several notions of distance metrics between distributions.  
\begin{defn}
If $P,Q$ are distributions, then,
\begin{itemize}
\item the \emph{statistical distance} or the \emph{total-variation distance} is $\SD(P,Q) = \frac12 \sum_{x} | P(x) - Q(x) |$, 
\item the \emph{$\CSD$-divergence} is $\CS(P \| Q) = \sum_{x} \frac{(P(x) - Q(x))^2}{Q(x)}$, and
\item the \emph{KL-divergence} is $\KL(P \| Q) = \sum_{x} P(x) \log \frac{P(x)}{Q(x)}$.
\end{itemize}
\end{defn}

\noindent Next, we have the following bound on the $\CSD$-divergence
between two Bernoulli distributions.
\begin{lem}
    [$\CSD$-Divergence between Bernoulli Distributions]
    \label{lem:chi-ub}
    For Bernoulli distributions $P$ and $Q$ over $\zo$
    (with means $p$ and $q$, respectively), such that
    $\abs{p-q} \leq \tfrac{1}{4}$ and
    $p \leq \tfrac{1}{2}$,
    $$\dcs(P,Q) \leq \frac{4(p-q)^2}{q}.$$
\end{lem}
\begin{proof}
    The proof is identical to that of Claim~5.12 from \cite{KamathLSU19}.
\end{proof}

\noindent For product distributions $P = P_1 \otimes \dots \otimes P_k$ and $Q = Q_1 \otimes \dots \otimes Q_k$, the KL-divergence is additive, and the total-variation distance and the $\CSD$-divergence are sub-additive. In particular, we have the following.
\begin{lem}
    [Sub-Additivity under Product Distributions]
    \label{lem:sd-ub}
    Let $P = P_1 \otimes \dots \otimes P_k$ and
    $Q = Q_1 \otimes \dots \otimes Q_k$ be two
    product distributions. Then,
    \begin{itemize}
    \item $\sd{P}{Q} \leq \sum_{j=1}^{d} \SD(P_j,Q_j),$
    \item $\CS(P \| Q) \leq \sum_{j=1}^{d} \CS(P_j \| Q_j)$, and
    \item $\KL(P \| Q) = \sum_{j=1}^{d} \KL(P_j \| Q_j).$
    \end{itemize}
\end{lem}

\noindent The three metrics are related to each other in a clean way as follows.
\begin{lem}[Pinsker's Inequality]\label{lem:pinsker}
For any two distributions $P$ and $Q$, we have,
$$2 \cdot \SD(P,Q)^2 \leq \KL(P \| Q) \leq \CS(P \| Q).$$
\end{lem}

\subsubsection{Tail Bounds}

We use a few tail bounds for sums of independent Bernoulli random variables.
The first lemma is an additive form of the Chernoff bound.
\begin{lem}[Bernstein's Inequality]\label{lem:chernoff-add}
    For every $p > 0$, if $X_1,\dots,X_m$ are i.i.d.\ samples from $\Ber(p)$,
    then for every $\eps > 0$,
    \begin{equation*}
        \pr{}{\frac{1}{m}\sum\limits_{i=1}^{m}{X_i} \geq
        	p + \eps} \leq e^{-\KL(p+\eps||p)\cdot m}
        ~~~\textrm{and}~~~\pr{}{\frac{1}{m}\sum\limits_{i=1}^{m}{X_i} \leq
            p - \eps} \leq e^{-\KL(p-\eps||p)\cdot m}
    \end{equation*}
\end{lem}

\noindent The second lemma is an multiplicative form of the Chernoff bound.
\begin{lem}[Multiplicative Chernoff Bound]\label{lem:chernoff-mult}
    For every $p > 0$, if $X_1,\dots,X_m$ are i.i.d.\ samples from $\Ber(p)$,
    then for every $\delta \geq 0$,
    \begin{equation*}
        \pr{}{\sum\limits_{i=1}^{m}{X_i} \geq
        	(1+\delta)pm} \leq e^{-\frac{\delta^2 pm}{2+\delta}}
    \end{equation*}
\end{lem}

\noindent The next lemma follows from Lemma~\ref{lem:chernoff-mult}, and bounds the norms of points sampled from a binary product distribution with bounded marginals.
\begin{lem}[Bounded Norms of Rows]\label{lem:chernoff-prod}
    Suppose $X_1,\dots,X_m$ are sampled i.i.d.\ from a product
    distribution $P$ over $\{0,1\}^t$, where the mean of each coordinate is upper
    bounded by $p$ (i.e., $\ex{}{P}\preceq p$). Then,
    \begin{enumerate}
    \item if $pt \geq 1$, then for each $i$,
        $\pr{}{\abs{X_i} \geq
    	pt\left(1+2\log(\frac{m}{\beta})\right)}
        \leq \frac{\beta}{m}$, and
    \item if $pt < 1$, then for each $i$,
        $\pr{}{\abs{X_i} \geq
    	4\log(\frac{m}{\beta})}
        \leq \frac{\beta}{m}.$
    \end{enumerate}
\end{lem}
\begin{proof}
    In the first case, we apply Lemma~\ref{lem:chernoff-mult} after setting $\delta=2\log(m/\beta)$ and $\mu = pm$, and by noting that $\log(m/\beta) \geq 1$. In the second case, we do the same as in the first case, but set $\delta=\tfrac{2\log(m/\beta)}{pm}$.
\end{proof}

\noindent Finally, we describe the concentration of Laplace random variables with mean $0$.
\begin{lem}[Laplace Concentration]\label{lem:lap-conc}
    Let $Z \sim \Lap(t)$. Then
    $\pr{}{\abs{Z} > t\cdot\ln(1/\beta)} \leq \beta.$
\end{lem}

\subsection{Privacy Preliminaries}

We start with the definition of differential privacy.
\begin{defn}[Differential Privacy (DP) \cite{DworkMNS06}]
    \label{def:dp}
    A randomized algorithm $M:\cX^n \rightarrow \cY$
    satisfies $(\eps,\delta)$-differential privacy
    ($(\eps,\delta)$-DP) if for every pair of
    neighboring datasets $X,X' \in \cX^n$
    (i.e., datasets that differ in at most one entry, denoted by $X \sim X'$),
    $$\forall~ Y \subseteq \cY,~~~
        \pr{}{M(X) \in Y} \leq e^{\eps}\cdot
        \pr{}{M(X') \in Y} + \delta.$$
    When $\delta = 0$, we say that $M$ satisfies
    $\eps$-differential privacy or pure differential
    privacy.
\end{defn}

\noindent These definitions of DP are closed under post-processing.
\begin{lem}[Post-Processing \cite{DworkMNS06}]\label{lem:post-processing}
    If $M:\cX^n \rightarrow \cY$ is
    $(\eps,\delta)$-DP, and $P:\cY \rightarrow \cZ$
    is any randomized function, then the algorithm
    $P \circ M$ is $(\eps,\delta)$-DP.
\end{lem}

\subsubsection{Known Differentially Private Mechanisms}

We state a standard result on achieving differential
privacy via noise addition proportional to the
\emph{sensitivity} of the function being computed~\cite{DworkMNS06}.
\begin{defn}[Sensitivity]
    Let $f : \cX^n \to \R^d$ be a function,
    its \emph{$\ell_1$-sensitivity} is
    $$\Delta_{f,1} \coloneqq \max_{X \sim X' \in \cX^n} \abs{f(X)-f(X')}.$$
\end{defn}

\noindent For a function with bounded $\ell_1$-sensitivity,
we can achieve $\eps$-DP by adding noise from
a Laplace distribution scaled to its
$\ell_1$-sensitivity.
\begin{lem}[Laplace Mechanism] \label{lem:laplacedp}
    Let $f : \cX^n \to \R^d$ be a function
    with $\ell_1$-sensitivity $\Delta_{f,1}$.
    Then the Laplace mechanism
    $$M(X) \coloneqq_R f(X) + \Lap\left(\frac{\Delta_{f,1}}
        {\eps}\right)^{\otimes d}$$
    satisfies $\eps$-DP, where ``$\coloneqq_R$'' is a notation we use to define a randomized mechanism.
\end{lem}

\noindent We finally state the result from \cite{HopkinsKMN23} about
estimating the means of sub-Gaussian distributions under
pure DP in polynomial time.
\begin{thm}[Sub-Gaussian Learner from {\cite[Theorem~5.1]{HopkinsKMN23}}]\label{thm:sg-learner}
    Assume that $0<\alpha,\beta,\eps<1$ and $R>0$.
    Let $\mu \in \R^d$, where $\|\mu\|_2 \leq R$, be unknown.
    There exists an $\eps$-DP algorithm ($\DPSGLearner$)
    that takes $n$
    i.i.d.\ samples from a sub-Gaussian distribution
    with mean $\mu$ and covariance $0 \preceq \Sigma \preceq \id$,
    such that,
    $$n \geq
        \wt{O}_{\alpha}\left(\frac{d + \log(1/\beta)}{\alpha^2}
            + \frac{d + \log(1/\beta)}{\alpha \eps}
            + \frac{d\log(R)}{\eps}\right),$$
    where $\wt{O}_{\alpha}(\cdot)$ hides polylogarithmic factors
    in $\tfrac{1}{\alpha}$, runs in time $\poly(n,d)$,
    and with probability at least $1-\beta$, outputs
    $\wh{\mu}$ such that $\|\mu - \wh{\mu}\|_2 \leq \alpha$.
\end{thm}

\section{A Pure DP Product Distribution Estimator}
\label{sec:product}

In this section we introduce and analyze our algorithm for learning a product distribution $P$ over $\zo^d$ in total-variation distance. The pseudocode is stated in Algorithm~\ref{alg:ppde}, with some additional observations and information regarding the notations stated in Section~\ref{sec:algorithm}. For simplicity of presentation, we assume that the product distribution has mean, whose marginals are bounded by $\frac12$ (i.e., $\ex{}{P} \preceq \frac12$), however, we would like to clarify that this assumption is essentially without loss of any generality, and is easily removable at the cost of a meagre constant factor in the sample complexity.
We make an additional assumption that
$d \geq 2$ and that $\beta \leq \tfrac{1}{2}$.

The following is the main result of our work. We prove the privacy, the computational efficiency, and the accuracy guarantees of our algorithm in Sections~\ref{sec:privacy}, \ref{sec:efficiency}, and~\ref{sec:accuracy}, respectively.
\begin{thm}\label{thm:ppde}
    For every $\eps,\alpha,\beta > 0$, there exists an $\eps$-DP algorithm ($\PPDE$) that takes $n$ i.i.d.\ samples from a product distribution $P$ over $\zo^d$, and returns a product distribution $Q$ over $\zo^d$ in $\poly(n,d,\tfrac{1}{\eps},\tfrac{1}{\alpha},\log(\tfrac{1}{\beta}))$ time, such that if
    $$n \geq \wt{O}_{\alpha}\left(
        \frac{d\log^2(d/\beta)}{\alpha^2} +
        \frac{d\log^2(d/\eps\beta)}{\eps\alpha}\right),$$
    where $\wt{O}_{\alpha}(\cdot)$ hides polylogarithmic factors
    in $\tfrac{1}{\alpha}$,
    then with probability at least $1-\beta$, $\SD(P,Q) \leq \alpha$.
\end{thm}

\subsection{The Algorithm}\label{sec:algorithm}

We first introduce a notation for the \emph{truncated mean} of all the points in a dataset. Given a data point $x \in \zo^{d}$ and $B \geq 0$, we write
\begin{equation*}
\trunc_{B}(x) = 
\begin{cases}
x &\textrm{if $\abs{x} \leq B$}\\
\frac{B}{\abs{x}} \cdot x &\textrm{if $\abs{x} > B$}
\end{cases}
\end{equation*}
to denote the truncation (or clipping) of $x$ to an $\ell_1$-ball of radius $B$.  Given a dataset $X = (X_1,\dots,X_{m}) \in \zo^{m \times d}$ and $B > 0$, we use
\begin{equation*}
\tmean_{B}(X) = \frac{1}{m} \sum_{i=1}^{m} \trunc_{B}(X_i)
\end{equation*}
to denote the mean of the truncated data points in $X$.
If one of the data points in $X$ does not satisfy
the norm bound (i.e., its $\ell_1$ norm is greater than $B$) when $X$ is served as an input to $\tmean_B$, then we will say,
``truncation occurred,'' as a shorthand.
We point out a couple of important observations.
\begin{itemize}
    \item The $\ell_1$-sensitivity of $\tmean_{B}$ is $\frac{B}{m}$, while the $\ell_1$-sensitivity of the un-truncated mean is $\tfrac{d}{m}$.
    \item Unless $\abs{X_{i}} > B$ for some $i \in [m]$, the truncated mean equals the un-truncated mean, i.e., $\tmean_{B}(X) = \frac{1}{m} \sum_{i=1}^{m} X_{i}$.
\end{itemize}

We also use the following notational conventions.
Given a data point $X_i \in \zo^d$, we use $X_i[j]$ to refer to its $j$-th coordinate, and for a subset of coordinates $S \subseteq [d]$, the notation $X_i[S] = (X_i[j])_{j \in S}$ to refer to the vector $X_i$ with coordinates restricted to $S$. Given a dataset $X = (X_1,\dots,X_m)$, we use the notation $X[S] = (X_1[S],\dots,X_m[S])$ to refer to the dataset consisting of each $X_i[S]$.
Next, for a domain $\cX$, a variable $x \in \cX$, and a probability distribution $D$ over $\cX$, we write ``$x \gets_R D$'' to indicate that a sample has been drawn from $D$, and assigned to $x$.
Finally, in Algorithm~\ref{alg:ppde}, we use $C_{\alpha}$ to denote the $\polylog(1/\alpha)$ quantity in the sample complexity from Theorem~\ref{thm:sg-learner}.

We first briefly describe Algorithm~\ref{alg:ppde} here. The algorithm runs in two phases, essentially -- partitioning and final phases.
\begin{itemize}
    \item \textbf{Partitioning:} This is an iterative phase. In each iteration, the algorithm truncates all the points in the dataset according to an upper bound, and computes the empirical mean of all the truncated points, and then adds independent Laplace noise to each of the remaining coordinates of the mean vector in that iteration, and works with only the noisy values of those coordinates. The coordinates with noisy values above a certain threshold are put aside, while those with noisy values below that threshold are the remaining coordinates to work with in the next iteration or in the final phase. The important observation is that the noisy values that appear large cannot actually be small without the noise with high probability, and vice versa, which we show later in the analysis. This helps us separate out the large coordinates with high confidence. In the next iteration, we reduce the upper bound and the threshold because the larger coordinates were separated out already, and we are now left with the ones with lower magnitudes. When the number of remaining coordinates is small enough, we exit the loop. At this point, we have all the batches of ``similar'' coordinates partitioned, so we scale the coordinates of all the batches up according to their respective upper bounds that we had used in the previous iterations in which they were separated out, so that they all now have magnitudes that are within a constant factor of one another. Then we apply the learner from Theorem~\ref{thm:sg-learner}, which gives us an accurate estimate of that scaled vector in $\ell_2$ distance. We rescale the coordinates of that estimate with the inverse of their respective, original scaling factors, which gives us an accurate total-variation estimate of those coordinates.
    \item \textbf{Final:} This is a ``one-shot'' phase. The coordinates, which were not estimated in the previous phase, are estimated in this phase. Here, the algorithm simply truncates all the points in the dataset as per a small upper bound, and computes their empirical mean, and then adds independent Laplace noise to all these remaining coordinates of the mean vector. The important observation here is that because the means of all these coordinates are small enough, we do not end up requiring a lot of noise to ensure privacy, so we are able to obtain an accurate $\ell_1$ estimate, which is sufficient to get an accurate total-variation estimate, at a low cost in the sample complexity.
\end{itemize}
In the end, the estimates from both the phases are combined appropriately, and released.

\begin{algorithm}[h!] \label{alg:ppde}
\caption{Pure DP Product Distribution Estimator $\PPDE_{\eps, \alpha, \beta}(X)$}
\KwIn{Samples $X_1,\dots,X_n \in \zo^{d}$ from an unknown product distribution $P$ satisfying $\ex{}{P} \preceq \frac12$.  Parameters $\eps, \alpha, \beta > 0$.}
\KwOut{A product distribution $Q$ over $\{0,1\}^d$ such that $\SD(P,Q) \leq \alpha$.}\vspace{10pt}
Set parameters:
$
R \gets \log_{2}\left(d/2\right)~~~
m \gets 2048 d \log(d/\beta) + \tfrac{2048 d \log(d/\eps\beta)}{\eps}
$\\
$
\qquad \qquad \qquad~~ m_0 \gets C_{\alpha}\left(\tfrac{d+\log(1/\beta)}{\alpha^2} + \tfrac{d+\log(1/\beta)}{\alpha\eps} + \tfrac{d\log(d)}{\eps}\right)~~~
m_1 \gets \tfrac{128d\log(d/\beta)}{\alpha^2} + \tfrac{256d\log(d/\eps\alpha\beta)}{\eps\alpha}
$\\
\vspace{10pt}

Split $X$ into three datasets $Y$, $Y^F$, and $Z$, of sizes $mR$, $m_1$, and $m_0$, respectively.\\
Split $Y$ into $R$ blocks $Y^1,\dots,Y^R$ of sizes $m$ each, denoted by $Y^{r} = (Y^{r}_{1},\dots,Y^{r}_{m})$.\\
Let $q \in (0,1)^d$ with $q[j] \gets 0$ for every $j \in [d]$, and let $S_1 = [d]$, $u_1 \gets \frac{1}{2}$, $\tau_1 \gets \frac{3}{16}$, and $r \gets 1$.\\
\vspace{10pt}

\tcp{Partitioning rounds.}
Let $T_P \gets \emptyset$.\\
\While{$u_r \abs{S_r} \geq 1$ and $r \leq R$}{
    Let $S_{r+1}, T_r \gets \emptyset$.\\
    Let $B_r \gets 3 u_r |S_r| \log(mR/\beta)$.\\
    Let $z_r \getsr \Lap\left(\tfrac{B_r}{\eps}\right)^{\otimes \abs{S_r}}$ and $q_r[S_r] \gets \tmean_{B_r}(Y^r[S_r]) + z_r$.\\
    \For{$j \in S_r$}{
        \eIf{$q_r[j] < \tau_r$}{
            Add $j$ to $S_{r+1}$.
            }
            {
            Add $j$ to $T_{r}$.
            }
    }
    Set $Z[T_r] \gets \tfrac{1}{\sqrt{u_r}}\cdot Z[T_r]$.\\
    Set $T_P \gets T_P \cup T_r$.\\
    \tcp{Update the loop's parameters.}
    Set $u_{r+1} \gets \frac12 u_{r}$, $\tau_{r+1} \gets \frac12 \tau_{r}$, and $r \gets r+1$.
}
\tcp{Run the sub-Gaussian learner from \cite{HopkinsKMN23} restricted to $T_P$.}
Let $\wh{q}[T_P] \getsr \DPSGLearner_{\eps,\frac{\alpha}{5},\beta,\sqrt{d}}(Z[T_P])$.\\
\For{$i \in [r-1]$}{
    Set $q[T_i] \gets \sqrt{u_i}\cdot\wh{q}[T_i]$.
}
\vspace{10pt}
    
\tcp{Final round.}
Let $S_F \gets [d] \setminus T_P$.\\
\If{$\abs{S_F} \geq 1$}{
    Let $B_F \gets 4\log(m_1/\beta)$.\\
    Let $z \getsr \Lap\left(\tfrac{B_F}{\eps}\right)^{\otimes \abs{S_F}}$ and $q[S_F] \gets \tmean_{B_F}(Y^F[S_F]) + z$.
}
\vspace{10pt}

\tcp{Return the final estimate.}
\Return $Q = \Ber(q[1]) \otimes \dots \otimes \Ber(q[d])$.
\end{algorithm}

\subsection{Privacy Analysis}\label{sec:privacy}

The privacy analysis of Algorithm~\ref{alg:ppde} is based on the privacy guarantees of the Laplace mechanism and bounded sensitivity of the truncated mean, along with the privacy guarantees of $\DPSGLearner$ (Theorem~\ref{thm:sg-learner}).

\begin{prop}\label{thm:ppde_dp}
    For every $\eps,\alpha,\beta > 0$, $\PPDE_{\eps,\alpha,\beta}(X)$ satisfies $\eps$-DP.
\end{prop}
\begin{proof}
    Each individual's data is used only once in
    exactly one of these three situations -- while computing
    $\tmean_{B_r}(Y^r)$ in some round $r$ of the partitioning
    rounds, in the call to $\DPSGLearner$ at the end of
    the \textbf{While}-loop of the partitioning rounds,
    or in the final round while computing $\tmean_{B_F}(Y^F)$.
    In other words, since all the datasets -- $Y^r$ (for all
    $i \in [R]$), $Y^F$, and $Z$ -- are disjoint and only
    used once in the entire algorithm, we do not need
    to apply composition, but instead, we just have to
    show that each individual computation is $\eps$-DP.

    In each partitioning round $r$, we perform an $\ell_1$
    truncation on all rows to within $B_r$, and add Laplace
    noise scaled to $\tfrac{B_r}{\eps}$ to every coordinate
    of $\tilde{p}_r[S_r]$. By Lemma~\ref{lem:laplacedp}, this
    satisfies $\eps$-DP. By similar reasoning, the final
    round also satisfies $\eps$-DP. Finally, the mean of
    $P$ has an $\ell_2$ norm of at most $\sqrt{d}$. Therefore,
    from the privacy guarantees of $\DPSGLearner$
    (Theorem~\ref{thm:sg-learner}), we have $\eps$-DP for
    this step, as well.
\end{proof}

\subsection{Efficiency Analysis}\label{sec:efficiency}

Here, we prove the computational efficiency of Algorithm~\ref{alg:ppde}, which is a key feature of our work.

\begin{prop}\label{thm:ppde_eff}
    There exists a polynomial $f:\R\times\R\to\R$, such that
    for every $\eps,\alpha,\beta,n,d>0$ and
    $X \in \zo^{n \times d}$,
    $\PPDE_{\eps,\alpha,\beta}(X)$ has a running time
    of $f(n,d)$.
\end{prop}
\begin{proof}
    There are at most $\log_2(d)$ iterations in the
    partitioning rounds, and in each iteration, we
    perform $O(d)$ operations. The call to $\DPSGLearner$
    after that costs another fixed polynomial
    $g(n,d)$ time, where $g:\R\times\R\to\R$
    (Theorem~\ref{thm:sg-learner}). The final round
    has an $O(d)$ running time, as well. Therefore,
    we have an $f(n,d) = O(d\log(d)) + g(n,d) + O(d)$
    running time, which is polynomial in $n,d$.
    Note that we are assuming that the cost in the
    running time due to sampling from Laplace
    distribution is very low.
\end{proof}

\subsection{Accuracy Analysis}\label{sec:accuracy}

In this section we prove the following proposition bounding the sample complexity required by \PPDE{} to be accurate.

\begin{prop}\label{thm:ppde_acc}
    For every $d \in \N$, every product distribution
    $P$ over $\zo^{d}$, and every $\eps,\alpha, \beta > 0$,
    if $X = (X_1,\dots,X_n)$ are independent samples
    from $P$, where
    $$
        n \geq \wt{O}_{\alpha}\left(\frac{d\log^2(d/\beta)}{\alpha^2}
        + \frac{d\log^2(d/\eps\beta)}{\alpha\eps}\right),
    $$
    then with probability at least $1-O(\beta)$,
    $\PPDE_{\eps,\alpha,\beta}(X)$ outputs $Q$,
    such that $\SD(P,Q) \leq \alpha$. The notation $\wt{O}_{\alpha}(\cdot)$ hides polylogarithmic factors in $\tfrac{1}{\alpha}$.
\end{prop}

\subsubsection{Analysis of the Partitioning Rounds}

In this section we analyze the progress made during the partitioning rounds.  We show two properties for any round $r$: (1) any coordinate $j$ such that $q_r[j]$ was filtered out during the partitioning rounds has a large mean, and (2) any coordinate $j$, such that $q_r[j]$ was moved on to the next round, has a small mean. We capture the properties of the partitioning rounds that will be necessary for the proof of Theorem~\ref{thm:ppde_acc} in the following lemma.

\begin{lem}[Partitioning Rounds]\label{lem:ppde_partitioning}
    If $Y^{1},\dots,Y^{R}$ each contain at least
    $$
    m \geq 2048 d \log(d/\beta)
        + \frac{2048 d \log(d/\eps\beta)}{\eps}
    $$
    i.i.d.\ samples from $P$, and $Z$ contains
    $$m_0 \geq \wt{O}_{\alpha}\left(\frac{d + \log(1/\beta)}{\alpha^2} + \frac{d + \log(1/\beta)}{\alpha\eps} + \frac{d\log(d)}{\eps}\right)$$
    i.i.d.\ samples from $P$ (where $\wt{O}_{\alpha}(\cdot)$ hides polylogarithmic factors in $\tfrac{1}{\alpha}$),
    then with probability at least $1-O(\beta)$,
    in every partitioning round $r \in [R]$, we have the following.
    \begin{enumerate}
        \item If a coordinate $j$ is filtered out in round $r$ (i.e., $q_r[j] \geq \tau_r$), then $p[j]$ is large: 
            $$
                p[j] \geq \tfrac{15\tau_r}{17}.
            $$
        \item If a coordinate $j$ is not filtered out in round $r$ (i.e., $q_r[j] < \tau_r$), then
            $p[j]$ is small:
            $$
                p[j] \leq u_{r+1} = \frac{u_{r}}{2}.
            $$
    \end{enumerate}
    Therefore, if $S_P \subseteq [d]$ is the set of all the coordinates estimated in the partitioning rounds, then with probability at least $1-O(\beta)$, $\SD(P[S_P],Q[S_P]) \leq \tfrac{\alpha}{2}$.
\end{lem}
\begin{proof}
We will prove the lemma by induction on $r$.
Therefore, we will assume that
in every round $r$, $p[j] \leq u_{r}$ for
every $j \in S_{r}$ and prove that if this
bound holds, then the two claims in the
lemma hold. For the base of the induction,
observe that, by assumption, $p[j] \leq u_{1} = \frac12$
for every $j \in S_{1} = [d]$.  In what
follows we fix an arbitrary round $r \in [R]$.
Throughout the proof, we will use the notation
$\tilde{p}_{r} = \frac{1}{m} \sum_{i = 1}^{m} Y^{r}_{i}$
to denote the empirical mean of the $r$-th
block of samples.

\begin{clm}[Sampling Error in Partioning Rounds]\label{clm:part-samp-err}
    If $\tilde{p}_r[j] = \frac{1}{m} \sum_{i=1}^{m} Y^{r}_{i}[j]$ and
    $m \geq 1024d\log(dR/\beta)$, then
    with probability at least $1-\frac{2\beta}{R}$,
    $$
        \forall j \in S_{r}~~\left| p[j] - \tilde{p}_r[j] \right|
        \leq \sqrt{\frac{4p[j]
        \log\left( \frac{dR}{\beta} \right)}{m}}.
    $$
    Furthermore, if $p[j] \geq \frac{1}{d}$, then
    $$\left| p[j] - \tilde{p}_r[j] \right|
        \leq \frac{p[j]}{16}.$$
\end{clm}
\begin{proof}
    We use Lemma~\ref{lem:chernoff-add} and the facts that for all $\gamma>0$ and $0 < p \leq 1$,
    $$\KL(p+\gamma||p) \geq \frac{\gamma^2}
        {2(p+\gamma)}~~~\textrm{and}~~~\KL(p-\gamma||p)
        \geq \frac{\gamma^2}{2p},$$
    and set
    $$\gamma = \sqrt{\frac{4p[j] \log \left(\frac{dR}{\beta}\right)}{m}}.$$
    Note that when $p[j] \geq \frac{1}{d}$, due
    to our choice of parameters, $\gamma \leq \tfrac{p[j]}{16}$.
    Therefore, $2(p[j]+\gamma) \leq 4p[j]$.
    Finally, taking a union bound over the cases
    when $\tilde{p}_r[j] \leq p[j] - \gamma$ and when
    $\tilde{p}_r[j] \geq p[j] + \gamma$, we
    prove the claim.
\end{proof}

\begin{clm}[No Truncation in Partioning Rounds]\label{clm:part-no-trunc}
    In round $r$, with probability at least $1-\frac{\beta}{R}$,
    for every $Y^{r}_{i} \in Y^{r}$, we have that
    $\abs{Y^{r}_{i}} \leq B_{r}$.
    So, no rows of $Y^{r}$ are truncated while
    computing $\tmean_{B_{r}}(Y^{r})$.
\end{clm}
\begin{proof}
    By assumption, all marginals specified
    by $S_r$ are upper bounded by $u_r$. Now,
    the expected value of $\abs{X^r_i}$ is
    at most $u_r|S_r|$. Since
    $B_r = 3 u_r |S_r| \log(mR/\beta)$,
    we know that $B_r \geq u_r |S_r| \left(1+2\log(mR/\beta)\right)$.
    The claim now follows from Lemma~\ref{lem:chernoff-prod}
    and a union bound over the rows of $Y^r$.
\end{proof}
    
\begin{clm}[Error due to Privacy in Partioning Rounds]\label{clm:part-priv-err}
With probability at least $1-\frac{2\beta}{R}$,
$$
    \forall j \in S_{r}~~\left| \tilde{p}_r[j] - q_{r}[j] \right| \leq
	\frac{3 u_r |S_r|
	\log\left(\frac{mR}{\beta}\right)
        \log\left(\frac{dR}{\beta}\right)}{\eps m}.
$$
\end{clm}
\begin{proof}
    We assume that all marginals specified
    by $S_r$ are upper bounded by $u_r$. From
    Claim~\ref{clm:part-no-trunc}, we know that,
    with probability at least $1-\beta/R$, there is no truncation,
    so $\tmean_{B_{r}}(Y^r[S_r]) = \frac{1}{m}
    \sum_{Y_i^{r} \in Y^{r}} Y^{r}[S_r] = \tilde{p}_r[S_{r}]$.
    So, the Laplace noise
    is added to $\tilde{p}_r[j]$ for each $j \in S_r$.
    Therefore, the only source of error here is the
    Laplace noise. Using the standard tail bound
    for Laplace distributions (Lemma~\ref{lem:lap-conc})
    after setting,
    $$
    t = \frac{3 u_r |S_r|
	\log\left(\frac{mR}{\beta}\right)}{\eps m},$$
    and taking a union bound over all
    coordinates in $S_r$, and the event
    of truncation, we obtain the claim.
\end{proof}

Applying the triangle inequality and simplifying, using our choice of $m$, and noting that in our algorithm, $\tau_r = \tfrac{3u_r}{8}$,
we get that in each round $r$, with probability at least $1-\tfrac{4\beta}{R}$,
\begin{align}
\abs{p[j] - q_r[j]} &\leq \sqrt{\frac{4p[j]
            \log\left( \frac{dR}{\beta} \right)}{m}} +
            \frac{3 u_r |S_r|
	   \log\left(\frac{mR}{\beta}\right)
            \log\left(\frac{dR}{\beta}\right)}{\eps m}\nonumber\\
        &\leq \frac{p[j]}{16} + \frac{3u_r}{128}\label{eq:part-err1}\\
        &= \frac{p[j]}{16} + \frac{\tau_r}{16}.\label{eq:part-err2} 
\end{align}
To simplify our calculations, we will define $e_{r,j}$ to be the quantity on the right-hand-side of Inequality~\ref{eq:part-err2}.

\begin{clm}[Noisy Estimates above Threshold]\label{clm:part-above-thresh}
    With probability at least
    $1-\frac{4\beta}{R}$,
    for every $j \in S_{r}$,
    $$q_{r}[j] \geq \tau_{r} \implies
        p[j] \geq \frac{15\tau_r}{17}.$$
\end{clm}
\begin{proof}
    We know that $\abs{p_j-q_r[j]} \leq e_{r,j}$ with
    high probability, which implies that $p[j] \geq q_r[j]-e_{r,j}$.
    Now, given that $q_r[j] \geq \tau_r$ and the bound on $e_{r,j}$ from Inequality~\ref{eq:part-err2}, we have the following.
    $$p[j] \geq q_r[j] - e_{r,j} \geq \tau_r - \frac{p[j]}{16} - \frac{\tau_r}{16} \implies p[j] \geq \frac{15\tau_r}{17}$$
    This completes the proof.
\end{proof}

\begin{clm}[Noisy Estimates below Threshold]\label{clm:part-below-thresh}
    With probability at least $1-\frac{4\beta}{R}$,
    for every $j \in S_{r}$,
    $$
    q_{r}[j] < \tau_{r} \Longrightarrow p[j] \leq u_{r+1} = \frac{u_{r}}{2}.
    $$
\end{clm}
\begin{proof}
    We know that with high probability,
    $p_j \leq q_r[j] + e_{r,j}$. But since
    $q_r[j] < \tau_j$, we know that
    $p[j] < \tau_r + e_{r,j}$. Also,
    we set $\tau_r = \frac{3}{4}u_{r+1}$ in our algorithm.
    Using the bound on $e_{r,j}$ from Inequality~\ref{eq:part-err1}, this gives us the following.
    $$p[j] < \tau_r + e_{r,j} \leq \frac{3u_{r+1}}{4} + \frac{p[j]}{16} + \frac{3u_r}{128} = \frac{3u_{r+1}}{4} + \frac{p[j]}{16} + \frac{3u_{r+1}}{256} \implies p[j] \leq \frac{13u_{r+1}}{16} < u_{r+1} = \frac{u_r}{2}$$
    Our proof is complete.
\end{proof}

Claim~\ref{clm:part-below-thresh} completes the
inductive step of the proof. It establishes
that at the beginning of round $r+1$, $p_j \leq u_{r+1}$
for all $j \in S_{r+1}$.

Finally, we analyse the accuracy of the coordinates
collected in $T_P$ over all the partitioning rounds
by the call to $\DPSGLearner$.
We can bound the total-variation distance between $P[S_P]$
and $Q[S_P]$ by computing the $\chi^2$ divergence between
the two. We make two key observations here.
\begin{itemize}
    \item In round $r$, we scale $Z[T_r]$ by $\tfrac{1}{\sqrt{u_r}}$, which means that the (scaled product) distribution of that subset of the coordinates (denoted by $\wh{P}[T_r]$) is over $\left\{0,\tfrac{1}{\sqrt{u_r}}\right\}$. Let $\Sigma_{T_r}$ be its covariance matrix. Then we can observe that the eigenvalues of $\Sigma_{T_r}$ lie between $\tfrac{\tau_r(1-\tau_r)}{u_r} = \tfrac{3(1-\frac{3u_r}{8})}{8} \geq \tfrac{39}{64}$ and $\tfrac{u_r(1-u_r)}{u_r} = 1-u_r \leq 1$. Since this is true for any $r$, this must be true for the entire $\wh{P}[T_P]$, as well.
    \item Let the output product distribution from invoking $\DPSGLearner$ on $Z[T_P]$ be $\wh{Q}[T_P]$ (which has mean $\wh{q}[T_P]$). Suppose the covariance matrix of $\wh{P}[T_P]$ is $\Sigma_{P_P}$, and the value of $r$ before the update in the end of the final iteration of the \textbf{While}-loop was $r^*$. Then from Theorem~\ref{thm:sg-learner} (accuracy guarantees of $\DPSGLearner$) and our setting of the accuracy parameters in the call to $\DPSGLearner$ in Algorithm~\ref{alg:ppde}, with probability at least $1-\beta$, the squared $\ell_2$ distance between $\wh{P}[T_P]$ and $\wh{Q}[T_P]$ is given by,
    \begin{align*}
        \frac{\alpha^2}{25} &\geq \|\wh{q}[T_P]-\wh{p}[T_P]\|_2^2\\
            &= \sum\limits_{i \in [r^*]}{\|\wh{q}[T_i]-\wh{p}[T_i]\|_2^2}\\
            &= \sum\limits_{i \in [r^*]}{\frac{\|q[T_i]-p[T_i]\|_2^2}{u_i}}\\
            &= \sum\limits_{i \in [r^*]}{\frac{3\|q[T_i]-p[T_i]\|_2^2}{8\tau_i}}\tag{$\tau_i = \tfrac{3u_i}{8}$}\\
            &\geq \sum\limits_{i \in [r^*]}{\sum\limits_{j \in T_i}{\frac{45\cdot(q[j]-p[j])^2}{136\cdot p[j]}}}\tag{Claim~\ref{clm:part-above-thresh}}\\
            &= \frac{45}{544}\cdot\sum\limits_{j \in T_P}{\frac{4(q[j]-p[j])^2}{p[j]}}.
    \end{align*}
    On rearranging the above, this implies that,
    $$\sum\limits_{j \in T_P}{\frac{4(q[j]-p[j])^2}{p[j]}} \leq \frac{\alpha^2}{2}.$$
\end{itemize}
The above, on combining with Lemmata~\ref{lem:chi-ub}, \ref{lem:sd-ub}, and~\ref{lem:pinsker}, implies that $\dtv(P[T_P],Q[T_P]) \leq \tfrac{\alpha}{2}$.

Now, we can take the union bound over all the failure events in all the rounds and over the failure of $\DPSGLearner$, so that the conclusions of Lemma~\ref{lem:ppde_partitioning} hold with probability $1-5\beta$.
This completes the proof.
\end{proof}

\subsubsection{Analysis of the Final Round}

In this section we show that the TV error of the coordinates $j$, such that $q[j]$ was set in the final round, is small.

\begin{lem}[Final Round]\label{lem:ppde_final}
    In the final round, let $k \in [R+1]$ for which
    $u_k |S_k| < 1$. If $p[j] \leq u_k$
    for every $j \in S_{F}$, and $Y^{F}$
    contains at least
    $$
    m_1 =  \frac{128 d \log (d/\beta )}{\alpha^2} +
        \frac{256 d \log(d/\eps\alpha\beta)}{\alpha\eps}
    $$
    i.i.d.\ samples from $P$, then with
    probability at least $1-O(\beta)$,
    then $\SD(P[S_F], Q[S_F]) \leq \tfrac{\alpha}{2}$.
\end{lem}
\begin{proof}
    Again, we use the notation,
    $\tilde{p} = \frac{1}{m_1} \sum_{i = 1}^{m_1} Y^{F}_{i}$,
    for the rest of this proof. First, we have
    two claims that bound the difference
    between $p[j]$ and $\tilde{p}[j]$.
    
    \begin{clm}[Sampling Error for Large Coordinates in Final Round]\label{clm:final-samp-err-lrg}
        For each $j \in S_F$, such that $p_j > \frac{1}{d}$,
        with probability at least $1-2\beta/d$,
        we have,
        $$
        \left| p[j] - \tilde{p}[j] \right|
            \leq \sqrt{\frac{4p[j] \log\left( \frac{d}{\beta} \right)}{m_1}}.
        $$
    \end{clm}
    \begin{proof}
        The proof is identical to that of Claim~\ref{clm:part-samp-err}.
    \end{proof}
    
    \begin{clm}[Sampling Error for Small Coordinates in Final Round]\label{clm:final-samp-err-sml}
        For each $j \in S_F$, such that $p_j \leq \frac{1}{d}$,
        with probability at least $1-2\beta/d$,
        we have,
        $$
        \left| p[j] - \tilde{p}[j] \right|
        	\leq \frac{\alpha}{8d}.
        $$
    \end{clm}
    \begin{proof}
        We use Lemma~\ref{lem:chernoff-add} and
        facts that 
        $$\forall \gamma > 0~~~\KL(p+\gamma||p) \geq \frac{\gamma^2}
        {2(p+\gamma)}~~~\textrm{and}~~~\KL(p-\gamma||p)
        \geq \frac{\gamma^2}{2p}.$$
        Set $\gamma = \tfrac{\alpha}{8d}$. Then from
        Lemma~\ref{lem:chernoff-add} and our choice of $m_1$, we have the following.
        \begin{itemize}
            \item For all $j \in S_F$, with probability at least $1-\tfrac{\beta}{d}$, $\tilde{p}[j] \leq p[j] + \gamma$.
            \item For all $j \in S_F$, with probability at least $1-\tfrac{\beta}{d}$, $\tilde{p}[j] \geq \max\{0,p[j]-\gamma\}$.
        \end{itemize}
        Applying the union bound, we get the required result.
    \end{proof}

    \begin{clm}[No Truncation in Final Round]\label{clm:final-no-trunc}
        With probability at least $1-\beta$,
        for every $Y^{F}_{i}
        \in Y^{F}$, $\abs{Y^{F}_{i}} \leq B_{F}$,
        so no rows of $Y^{F}$ are truncated in the
        computation of $\tmean_{B_{F}}(Y^{F})$.
    \end{clm}
    \begin{proof}
        Note that all the marginals specified
        by $S_F$ are upper bounded by $u_k$ (where
        $k$ is the index as specified in
        Lemma~\ref{lem:ppde_final}) and that $u_k|S_F| < 1$.
        With this, we use Lemma~\ref{lem:chernoff-prod}
        and get the required result because we set
        the truncation radius $B_F = 4\log(m_1/\beta)$.
    \end{proof}
    
    \begin{clm}[Error due to Privacy in Final Round]\label{clm:final-priv-err}
    With probability at least $1-2\beta$,
    $$
    \forall j \in S_{r}~~\left| \tilde{p}[j] - q_{r}[j] \right| \leq
    \frac{4\log\left(\frac{m_1}{\beta}\right)\log\left(\frac{d}{\beta}\right)}{\eps m_1}.
    $$
    \end{clm} 
    \begin{proof}
        Using the standard tail bound for Laplace random
        variables (Lemma~\ref{lem:lap-conc}) with the
        following parameters,
        $$
            t = \frac{4\log\left(\frac{m_1}{\beta}\right)}{\eps m_1},
        $$
        and taking the union bound over all the
        columns of the dataset in that round and
        the event of truncation, we obtain the claim.
    \end{proof}

    Note that because no truncation happens in this round,
    the sampling error without the Laplace noise is
    bounded by the quantity as specified in
    Claims~\ref{clm:final-samp-err-lrg}
    and~\ref{clm:final-samp-err-sml}. Because of the
    possible difference in magnitudes of the means
    of the marginals in $S_F$, we establish their error
    guarantees separately. Let $H \subset S_F$ be the
    set of all coordinates, whose means are greater
    than $\tfrac{1}{d}$, i.e.,
    $H \coloneqq \{j \in S_F: p[j] > 1/d\}$.
    Likewise, let $L \subseteq S_F$ be the set of lighter
    coordinates, i.e., $L \coloneqq S_F \setminus H$.

    We analyse the heavier coordinates in $H$ first. Let
    $\wt{P}[H]$ be the product distribution over the
    coordinates in $H$ with mean $\tilde{p}[H]$.
    For all $j \in H$, using Claim~\ref{clm:final-samp-err-lrg}
    and our choice of $m_1$, we know that
    $$\abs{\tilde{p}[j]-p[j]} \leq \sqrt{\frac{4p[j] \log\left( \frac{d}{\beta} \right)}{m_1}} \implies \frac{4(p[j]-\tilde{p}[j])^2}{p[j]} \leq \frac{16\log\left(\frac{d}{\beta}\right)}{m_1} \leq \frac{\alpha^2}{32d},$$
    which (by Lemma~\ref{lem:chi-ub}) means that
    $\dcs(\wt{P}[j],P[j]) \leq \tfrac{\alpha^2}{32}$.
    Combined with Lemma~\ref{lem:pinsker}, this implies
    that $\dtv(P[j],\wt{P}[j]) \leq \tfrac{\alpha}{8d}$.
    Now, from Claim~\ref{clm:final-priv-err} and our choice
    of $m_1$, we know that for all $j \in H$,
    $$\abs{\tilde{p}[j]-q[j]} \leq \frac{4\log\left(\frac{m_1}{\beta}\right)\log\left(\frac{d}{\beta}\right)}{\eps m_1} \leq \frac{\alpha}{8d},$$
    which implies that
    $\dtv(\wt{P}[j],Q[j]) \leq \tfrac{\alpha}{8d}$.
    Therefore, by triangle inequality,
    $\dtv(P[j],Q[j]) \leq \tfrac{\alpha}{4d}$. Finally,
    from Lemma~\ref{lem:sd-ub}, we have that
    $\dtv(P[H],Q[H]) \leq \tfrac{\alpha}{4}$.

    Next, we bound the error on the lighter coordinates
    in $L$. Claims~\ref{clm:final-samp-err-sml}
    and~\ref{clm:final-priv-err}, the triangle inequality,
    and our choice of $m_1$ show that for all $j \in L$,
    \begin{align*}
        \abs{p[j]-q[j]} &\leq \frac{\alpha}{8d} + \frac{4\log\left(\frac{m_1}{\beta}\right)\log\left(\frac{d}{\beta}\right)}{\eps m_1}\\
            &\leq \frac{\alpha}{8d} + \frac{\alpha}{8d}\\
            &= \frac{\alpha}{4d}.
    \end{align*}
    This implies that for all $j \in L$,
    $\dtv(P[j],Q[j]) \leq \tfrac{\alpha}{4d}$.
    Therefore, from Lemma~\ref{lem:sd-ub},
    $\dtv(P[L],Q[L]) \leq \tfrac{\alpha}{4}$.

    Finally, through an application of Lemma~\ref{lem:sd-ub}
    again, we obtain that
    $$\dtv(P[S_F],Q[S_F]) \leq \dtv(P[H],Q[H]) + \dtv(P[L],Q[L]) \leq \frac{\alpha}{2}.$$
    This completes our proof.
\end{proof}

\subsubsection{Putting It All Together}

In this section, we combine Lemmata~\ref{lem:ppde_partitioning} and~\ref{lem:ppde_final} to prove Proposition~\ref{thm:ppde_acc}. First, by Lemma~\ref{lem:ppde_partitioning}, with probability at least $1-O(\beta)$, if $S_{P}$ is the set of coordinates $j$, such that $q[j]$ was set in any of the partitioning rounds, then, 
\begin{enumerate}
\item $\SD(P[S_{P}], Q[S_{P}]) \leq \tfrac{\alpha}{2}$ and
\item if $j \not\in S_{F}$ and $k$ is the index of the final round, then $p[j] \leq u_{k}$ and $u_{k}\abs{S_F} < 1$.
\end{enumerate}
Next, due to the second consequence listed above, we can apply Lemma~\ref{lem:ppde_final} to obtain that if $S_{F}$ consists of all coordinates set in the final round, then with probability at least $1-O(\beta)$, $\SD(P[S_{F}], Q[S_{F}]) \leq \tfrac{\alpha}{2}$. Finally, we use the union bound and Lemma~\ref{lem:sd-ub} to conclude that, with probability at least $1-O(\beta)$,
$$
    \SD(P,Q) \leq \SD(P[S_{P}], Q[S_{P}]) + \SD(P[S_{F}], Q[S_{F}]) \leq \alpha.
$$
This completes the proof of Proposition~\ref{thm:ppde_acc}.

\section{Discussion}

In this work, we solved a fundamental statistical problem of estimating the means of binary product distributions in total-variation distance under pure DP with optimal sample complexity and under polynomial running time.

However, we would like to mention again that our techniques hold similarities with those in \cite{KamathLSU19}. That said, we note that private preconditioning steps are commonly seen in DP statistics now, especially when trickier, direction-wise metrics, such as total-variation distance, are involved. The goal in these cases is to have direction-wise accuracy, so the choice of the error metric is crucial. That said, the way this preconditioning is done could depend on the family of distributions in question according to the way total-variation distance is characterised for those distributions and on other factors, such as concentration properties and domain. For example, in \cite{KamathLSU19} and \cite{KamathMSSU22}, such steps were performed for estimation of covariances of Gaussians – the idea was to make all the directions of the Gaussian similar to one another privately, and then estimate them accurately, before reverting the transformation. There, the preconditioning looked different from what we did here, but the high-level goal was still the same – adding appropriate amounts of noise in all directions. The problem being solved in our work is also under total-variation distance, which is why we went back to a preconditioning-style algorithm. Is there a different (but a more direct) approach to solving this problem efficiently under pure DP that did not involve any private preconditioning? We do not know the answer yet, but it is an interesting question to think about.

Now, the steps in our algorithm were motivated by the current tools available to solve this problem. We realised that the sub-Gaussian mean estimator from \cite{HopkinsKMN23} could not be directly applied to our problem, otherwise the problem would become quite trivial to solve. Since estimating in total-variation distance requires direction-wise accuracy, we had to adapt the preconditioning approach in \cite{KamathLSU19} for the estimator from \cite{HopkinsKMN23} to give us anything useful. This also led to requiring different technical lemmata at different stages to prove the accuracy guarantees of our algorithm, thereby creating important and non-trivial differences in our analyses from those in \cite{KamathLSU19}.

We also admit that the recent development of \cite{HopkinsKMN23} was an important factor in our work. Before that, there was no efficient, pure DP method to estimate the means of sub-Gaussians. However, as we mentioned above, this was not enough by itself because it can only give an estimate that is accurate to within $\ell_2$ distance, so our preconditioning approach seemed necessary to us if we were to use the algorithm from \cite{HopkinsKMN23} as a black box.

\section*{Acknowledgements}

We would like to thank Gautam Kamath for their helpful discussion on this problem.

\printbibliography


\end{document}